\documentclass[11pt]{amsart}
\usepackage{comment}
\usepackage[utf8]{inputenc}
\usepackage[T1]{fontenc}
\DeclareMathAlphabet{\mathpzc}{OT1}{pzc}{m}{it}
\usepackage{geometry}
\usepackage[dvipsnames]{xcolor}

\usepackage{xstring}
\newcommand{\bef}[2][e]{%
    \IfStrEqCase{#1}{%
        {e}{\textcolor{Blue}{\textbf{*Beckett: #2*}}}%
        {c}{}%
    }
}
\newcommand{\nd}[2][e]{%
    \IfStrEqCase{#1}{%
        {e}{\textcolor{PineGreen}{\textbf{*Dr. Nathan: #2*}}}%
        {c}{}%
    }

}
\newcommand{\mhadd}[2][e]{%
    \IfStrEqCase{#1}{%
        {e}{\textcolor{Blue}{\textbf{#2}}}%
        {r}{#2}%
    }
}
\newcommand{\ndadd}[2][e]{%
    \IfStrEqCase{#1}{%
        {e}{\textcolor{PineGreen}{\textbf{#2}}}%
        {r}{#2}%
    }
}

\usepackage{amsfonts}
\usepackage{amssymb}
\usepackage{amsthm}
\usepackage{amsmath}
\usepackage{amscd}
\usepackage[shortlabels]{enumitem}
\usepackage{mathrsfs}
\usepackage{tikz}
\usetikzlibrary{calc,arrows,decorations.pathreplacing}
\usepackage{nicefrac, xfrac}
\usepackage{mathtools,xparse}
\usepackage{soul}
\usepackage{appendix}

\usepackage[pagebackref, colorlinks = true,
linkcolor = red,
urlcolor  = blue,
citecolor = blue,
anchorcolor = blue]{hyperref}

\theoremstyle{plain}

\newtheorem{theorem}{Theorem}[section]

\newtheorem{proposition}[theorem]{Proposition}

\newtheorem{question}[theorem]{Question}

\newtheorem{exercise}[theorem]{Exercise}

\theoremstyle{definition}

\newtheorem*{theorem*}{Theorem}

\renewcommand{\epsilon}{\varepsilon}

\newcommand{\vertiii}[1]{{\left\vert\kern-0.25ex\left\vert\kern-0.25ex\left\vert #1
		\right\vert\kern-0.25ex\right\vert\kern-0.25ex\right\vert}}

\DeclareSymbolFont{bbold}{U}{bbold}{m}{n}
\DeclareSymbolFontAlphabet{\mathbbold}{bbold}

\newcommand{\NN}{\ensuremath{\mathbb N}}

\newcommand{\ZZ}{\ensuremath{\mathbb Z}} 

\providecommand{\phantomsection}{}
\AtBeginDocument{\let\textlabel\label}
\makeatletter
\newcommand{\mylabel}[2]{\raisebox{.7\normalbaselineskip}{\phantomsection}(#1)%
	\def\@currentlabel{#1}\textlabel{#2}}
\makeatother

\makeatletter
\newcommand\xlabel[2][]{\phantomsection\def\@currentlabelname{#1}\label{#2}}
\makeatother

\ExplSyntaxOn
\NewDocumentCommand{\mathlist}{ O{,} m m }
 {
  \egreg_mathlist:nnn { #1 } { #2 } { #3 }
 }

\seq_new:N \l__egreg_mathlist_seq
\cs_new_protected:Npn \egreg_mathlist:nnn #1 #2 #3
 {
  \seq_set_split:Nnn \l__egreg_mathlist_seq { #1 } { #3 }
  \seq_use:Nnnn \l__egreg_mathlist_seq { #2 } { #2 } { #2 }
 }
\ExplSyntaxOff

\newcommand{\card}[2]{\textbf{#1} \textit{(#2)}}
\newcommand{\jcard}[1]{\textbf{#1}}

\allowdisplaybreaks

\numberwithin{equation}{section}
\title[]{Riftbound is Turing Complete}
\date{\today}

\author[N. Dalaklis]{Nathan Dalaklis} \thanks{Acknowledgments: The authors thank the local Norman, OK \textit{Riftbound} community for discussions about the game which led to this work, and also to the wider \textit{Riftbound} community for catching an error in an initial version of this work. This research received no external funding, and the bulk of the work was completed in the summer of 2026 during a research experience for undergrads designed by the first author. This work has not been edited or generated in any part by AI.}
\address[N. Dalaklis]{Department of Mathematics, University of Oklahoma, Norman, OK 73019, USA}
\email{\href{ndalaklis@ou.edu}{ndalaklis@ou.edu}}

\author[B. Fields]{Beckett Fields}
\address[B. Fields]{Department of Mathematics, University of Oklahoma, Norman, OK 73019, USA}
\email{\href{Beckett.E.Fields-1@ou.edu}{Beckett.E.Fields-1@ou.edu}}

\begin{document}

\begin{abstract}
    \textit{Riftbound: League of Legends Trading Card Game} is a trading card game about capturing and holding locations in a king-of-the-hill style contest. Originally released in China in August of 2025, and later released in the United States in October of 2025, the game has been well received for its depth and complexity. In this paper we demonstrate a facet of this complexity by providing sequences of valid game states which construct Universal Turing machines within the game. Each of these machines are constructed with tournament legal decks at the time of writing and strategies assigned are directed by the game state. We also show that given an appropriate board state the machine may be constructed and the computation may be performed in one game turn.
\end{abstract}
\maketitle

\section{Introduction} 
\textit{Riftbound:\;League of Legends Trading Card Game}, henceforth \textit{Riftbound}, is a trading card game developed by Riot Games and published by UVS Games and Shining Soul. Formally, the standard format of \textit{Riftbound} is a two-player zero-sum stochastic card game with imperfect information where decks are constructed by players from the published card-pool of just under $1000$ unique cards (at the time of writing) \cite{RDecks}. Other multiplayer formats also exist. In this way, \textit{Riftbound} is similar to \textit{Magic:\;The Gathering}, henceforth \textit{Magic}, which, to our knowledge, is the only trading card game known to be Turing complete prior to this work \cite{TCBlog,MagicCBH20}.

There are several multiplayer computer games that have mechanics built-in that readily allow for the implementation of Turing machines. Games like \textit{Minecraft} and \textit{Factorio} are such examples. As a result, the study of computational complexity in modern popular multiplayer games has focused on games where the existence of a construction is not obvious. Recently, the Turing completeness of \textit{Sid Meier's Civilization} was shown in \cite{CIV23}. In \cite{MagicPL24}, a programming language was embedded into Magic improving on the previously mentioned work on the game. In many descriptions of implementing Turing machines (TMs) and/or Universal Turing machines (UTMs) in modern multiplayer games we see that the implementation of the machine occurs over several turns and and in some cases requires an infinite turns assumption. In our case, once an appropriate board state is achieved, the machine can be constructed and run in a single turn of \textit{Riftbound}.

For trading card games in particular, \textit{Magic} has been investigated at length over the past decade (see again \cite{TCBlog,MagicCBH20,MagicPL24}). The description of the universal turing machine constructed in \textit{Magic} is run over multiple turns, and it enjoys the additional property that all actions of both players are forced. Although \textit{Riftbound} also has mechanics that force player action, we instead focus on the computational complexity of turn strategies and their interaction in the game of Riftbound. \footnote{For example, \card{Blind Fury}{OGN 025}, \card{Promising Future}{OGN 115}, and \card{Cull of The Weak}{OGN 209} force player action.} We prescribe a strategy whose implementation is directed by the game's board state and show that the resulting series of actions are equivalent to simulating computations on a given TM with a fixed number of states and alphabet symbols. As a result we recover that there exist strategies in \textit{Riftbound} for which the end state of a given turn in the strategy is undecidable.

\subsection{Main Theorems}
We call the machine constructed herein Jayce's Computer. A majority of the construction of the Turing machine and the computation on a given input occur solely on the constructor's sixth turn of the game in a kind of solitaire which sets our machine apart from the construction in \textit{Magic} and in other games where the machine takes several turns to implement. 

We also specify that for our setting we only consider game-state-directed strategies or strategies whose actions are only dependent on the overall state of the game. These strategies are natural for our context as they conform to how most players actually play trading card games. This restriction also removes trivially complex strategies from our analysis. For an example of one such excluded strategy, take a Turing machine input pairing $(M,w)$ and prescribe:
\begin{enumerate}
    \item Compute $M$ on $w$.
    \item If $M$ halts on $w$ concede
\end{enumerate}
The strategy herein is independent of the game state but deciding its outcome is equivalent to the halting problem. 

Given an arbitrary turing machine $M$ with $7$ states and $4$ symbols, Jayce's stratetgy will embed $M$ and will thus allow us to embed a member of class UTM$(7,4)$ which was shown to be non-empty in \cite[Section 6.]{RS96}. Since the computation of Jayce's Computer will occur on a single turn we also show that 

\begin{theorem}\label{MainTheorem}
    Whether or not the application of a game-state-directed strategy on a single turn of \textit{Riftbound} will result in a game win is undecidable. 
\end{theorem}

\subsection{Structure of the Paper}

We briefly recall preliminaries regarding Turing machines and our working definitions in Section \ref{Prelim}. We then detail the construction throughout Section \ref{JC} and its subsections. Subsections \ref{JS} and \ref{JCSetup} describe a route of valid game states one may take to enter a board state were construction of a TM with $7$ states and a $4$ symbol alphabet and its computation are possible. In Subsection \ref{CR&TS} we detail subroutines of the strategy necessary for continuing the turn and instantiating the components necessary for representing the TM. In Subsection \ref{Transitions} we detail the subroutines of the strategy used for implementing the transition function on Jayce's Computer and how we may encode halting. In Section \ref{strategy}, we explicitly state the strategy that proves Theorem \ref{MainTheorem}, we also comment on the complexity of said strategy here. In Section \ref{C&P} we briefly discuss the construction of other Turing machines in \textit{Riftbound}. Section \ref{Append} is an appendix containing a rough introduction to types of cards and concepts in \textit{Riftbound}. Section \ref{D&T} details the decklist used for Jayce's Computer in terms of where each card resides at the time construction of the machine is ready to begin.

\section{Preliminaries}\label{Prelim} A \textbf{Turing machine} is a tuple $M=(Q,\Sigma,\delta,q_0,q_{\text{accept}},q_{\text{reject}})$ where $Q$ is the set of states, $\Sigma$ is the set of tape symbols, $\delta: Q\times \Sigma \to Q\times \Gamma \times\{L,R\}$ is the transition function, and $q_0,q_{\text{accept}},q_{\text{reject}}\in Q$ are the start, accept and reject states respectively. For additional background on Turing machines we direct the reader to \cite[Chapters 3-5]{Sip13}. 
$\text{TM}(m,n) := \{M\::\: |Q| = m, |\Sigma| = n\}$ is the class of turing machines with $m$ states and an $n$ symbol alphabet. A Turing machine $M$ is called \textbf{universal} if it can simulate any other Turing machine. Such Turing machines were shown to exist by Turing in \cite{T37,T38}. The class UTM$(m,n)$ is the collection of universal turing machines with $m$ states and an $n$ symbol alphabet. For certain small values of $m$ and $n$, UTM$(m,n)$ was shown to be non-empty in \cite{RS96}. In an abuse of notation, henceforth we refer to a representative of this class by  UTM$(m,n)$ instead of naming one.

Let $X$ be a state space of a multiplayer turn-based game. A \textbf{game-state-directed strategy} is a function $f:X\to X$ whose iterates give an algorithm by which a player will progress the game-state on a given turn. A trading card game Turing machine is a machine-game-directed-strategy pair $(M,S)$ such that $S$ simulates $M$ in the game. In our setting of trading card game Turing machines (TCGTM), we say a TCGTM has \textbf{string-card complexity} $O(f(n))$ if the number of cards required to initialize any string of length $n$ on the tape is bounded above by $Nf(n)$ for some $N>1$. A \textbf{turn complexity} can be defined for multiplayer Turing machines (MPTM) that are deciders as well. That is, a decider MPTM has turn complexity $O(f(n))$ if for an input string $w$ of length $n$ the number of turns required to run the MPTM on $w$ is bounded above by $Nf(n)$ for some $N>1$. As we will see, Jayce's Computer will have turn complexity $O(1)$ when it is implemented with a decider. 

\section{Jayce's Computer}\label{JC}
In this construction we will embed an arbitrary $M\in\text{TM}(7,4)$. This will allow us to use UTM$(7,4)$ from \cite[Section 6.]{RS96} in our game-state-directed strategy for Jayce. Our first goal is to show that there exists a sequence of valid game states where one may construct the machine so there are several assumptions we may make to simplify Jayce's construction of the machine.
\subsection{Assumptions} \label{JS}
In the following construction we suppose that the opponent scores on a $(0,1,3,5,7,8)$ point tempo. We can increase the probability that opponent plays at this tempo by playing \card{Forgotten Monument}{SFD 209: Players can't score here until their third turn} as our battlefield. We also assume that the opponent does not alter Jayce's board state and does not score on Jayce's turn. Under this assumption, we may also assume that the opponent does not impede Jayce's ability to score, play spells, or move units. We address these concerns in more detail in Section \ref{C&P}. Without loss of generality we assume that Jayce goes first. In fact, to prove our main theorem, we need only demonstrate that there is a valid sequence of game states that implement a universal Turing machine in \textit{Riftbound} and that the computation on which happens on a single turn.
\subsection{Setup}\label{JCSetup} Here we detail one route through the deck to the game state where we can begin the construction of Jayce's Computer. This requires that Jayce draw the entirety of their deck without decking out and have the ability to generate an unbounded amount of energy and mind power. This process is inspired by \cite{JayceYT} however our deck is slightly different. We note that alternate routes are possible. We will first detail how to draw the entire deck and place the desired cards into the draw pile for initializing the resource loop. Then we describe the resource generation loop.
    \subsubsection{Jayce Turn 1} 
    
    \begin{enumerate}
    \item Play gear \card{Hextech Formula}{VEN 62: This enters exhausted./ Exhaust: Empower another gear.} for $2$ energy and pass.
    \end{enumerate}
    
    \subsubsection{Bob Turn 1} Under our assumptions Bob does not score and passes the turn. 
    \subsubsection{Jayce Turn 2 (Bob Score 0)} \label{JT2}
    \begin{enumerate}
        \item Play gear \card{Platewyrm Egg}{VEN 75: This enters exhausted./.../ Reaction: Exhaust: Add $1$ energy. if this is empowered, add $2$ energy} for $3$ energy.
        \item For $1$ energy activate legend \card{Jayce, Defender of Tomorrow}{VEN 149: Empower: 2 Energy 2 Power/ 1 Energy, Exhaust: Ready a gear. / Empowered, 1 Energy, Exhaust, Ready $2$ gear} ability targeting \jcard{Platewyrm Egg} to ready it. \label{JLA}
        \item Exhaust \jcard{Hextech Formula} targeting \jcard{Platewyrm Egg} to empower it. 
        \item Exhaust \jcard{Platewrym Egg} to add $2$ energy
        \item Play \jcard{Hextech Formula} for $2$ energy and pass the turn.
    \end{enumerate}
    \subsubsection{Bob Turn 2} Under our assumptions Bob scores a point and then passes the turn.
    \subsubsection{Jayce Turn 3 (Bob Score 1)} \label{JT3}
    \begin{enumerate} 
        \item Exhaust empowered \jcard{Platewrym Egg} to add $2$ energy. \label{JEPWE}
        \item Exhaust all runes and recycle a mind rune to add a mind power. Play unit \card{Jayce, Brilliant Inventor}{VEN 068: When you play me or the first time you play a non-token gear each turn, you may ready something besides me that's exhausted.} for $6$ energy and $1$ mind power to ready empowered \jcard{Platewrym Egg}. 
        \item Exhaust empowered \jcard{Platewrym Egg} to add $2$ more energy ($4$ energy remains) \label{JEPWEagain}
        \item Play gear \card{Sumpworks Map}{VEN 085: Reaction / Temporary / When an opponent scores, draw 1.} for $2$ energy ($2$ energy remains). On play, \jcard{Jayce, Brilliant Inventor} effect will trigger. Here target empowered \jcard{Platewyrm Egg} to ready it.
        \item Exhaust empowered \jcard{Platewrym Egg} to add two energy ($4$ energy remains).
        \item Repeat \ref{JT2}.(\ref{JLA}) and the previous step to add a net $1$ energy ($5$ energy remains).
        \item Use the remaining energy to play another \jcard{Sumpworks Map} for $2$ energy and a \card{Questionable Tome}{VEN 054: Empower: Exhaust / Disempower this, 1 Energy, Exhaust: Draw 1} for $3$ energy.
        \item Exhaust \jcard{Questionable Tome} to empower it. Then pass.
    \end{enumerate}
    \subsubsection{Bob Turn 3} Under our assumptions Bob scores two points, each time Bob scores, each copy of \jcard{Sumpworks Map} triggers. As a result Jayce will have drawn a total of $4$ cards this turn.
    \subsubsection{Jayce Turn 4 (Bob Score 3)}\label{JT4}
    \begin{enumerate}
        \item At the beginning phase, both copies of \jcard{Sumpworks Map} will be killed and go to the trash since they both have temporary. 
        \item Repeat \ref{JT3}.(\ref{JEPWE})-(\ref{JEPWEagain}) and exhaust the remaining rune to add a total of $5$ energy.
        \item Use $3$ energy to play another copy of \jcard{Platewyrm Egg}. Both copies of \jcard{Jayce Brilliant Inventor} will trigger, target both copies of \jcard{Platewrym Egg} to ready them both.
        \item Exhaust one of the copies of \jcard{Hextech Formula} to empower the copy of \jcard{Platewyrm Egg} that is not empowered.
        \item Exhaust both empowered copies of \jcard{Platewrym Egg} to add $4$ energy (total $6$ available) \label{EEPWE}
        \item Recycle $2$ runes to add $2$ power, use this and $2$ energy to empower \jcard{Jayce, Defender of Tomorrow}
        \item Spend $1$ energy and exhaust \jcard{Jayce, Defender of Tomorrow} to ready both empowered copies of \jcard{Platewrym Egg}. Repeat \ref{JT4}.(\ref{EEPWE}) to add $4$ energy (7 energy remains) \label{EEJL+EEPWE}
        \item Spend $1$ energy, disempower and exhaust \jcard{Questionable Tome} to draw $1$. (6 energy remains). \label{QT}
        \item Spend 3 energy and recycle a blue rune to play \card{Heimerdinger, Inventor}{OGN 111: I have all Exhaust abilities of all friendly legends, units, and gear.} (3 energy remains). 
        \item Spend $2$ energy to play the final copy of \jcard{Sumpworks map} (1 energy remains)
        \item Exhaust the second copy of \jcard{Hextech Formula} to empower \jcard{Questionable Tome} and pass.
    \end{enumerate}
    \subsubsection{Bob Turn 4} Under our assumptions Bob scores two points, each time Bob scores, \jcard{Sumpworks Map} triggers. As a result Jayce will have drawn a total of $2$ cards this turn.
    \subsubsection{Jayce Turn 5 (Bob Score 5)} \label{JT5}
    \begin{enumerate}
        \item At the beginning phase, \jcard{Sumpworks Map} will be killed and go to the trash since it has temporary. 
        \item Exhaust $4$ runes and spend the $4$ energy created to play \card{Applied Researchers}{VEN 055: Empower: 3 Energy / Empowered: Your spells cost $1$ Energy and $1$ Power less, to a minimum of $1$ Energy}. 
        \item Exhaust the remaining rune and repeat both \ref{JT4}.(\ref{EEPWE}) and \ref{JT4}.(\ref{EEJL+EEPWE}) to add a total of $8$ energy.
        \item Recycle a mind rune to add one mind power. Use this power to play gear \card{Seal of Insight}{OGN 120: Exhaust: Reaction, add $1$ mind power.} Both copies of \jcard{Jayce, Brilliant Inventor} trigger. Target the empowered copies of \jcard{Platewrym Egg} to ready them. (8 energy remains)
        \item Repeat \ref{JT4}.(\ref{EEPWE}) to add $4$ energy ($12$ energy remains).
        \item Use $10$ energy to empower the \jcard{Applied Researcher} already in base and to play and empower a second \jcard{Applied Researcher}. ($2$ energy remains).
        \item Exhaust \jcard{Heimerdinger, Inventor}, copying \jcard{Questionable Tome}'s exhaust ability, to empower \jcard{Heimerdinger}.
        \item Exhaust \jcard{Seal of Insight} to add $1$ mind power. ($2$ energy and $1$ mind power remain) \label{AddMindPower}
        \item Due to the discount of \jcard{Applied Researchers} spend $1$ energy and the mind power to play \card{Premonition}{SFD O87: Reaction/ Draw $3$}. ($1$ energy remains) \label{PlayPrem}
        \item Due to the discount of \jcard{Applied Researchers} spend the last energy to play \card{Acceleration Gate}{VEN 150: Ready up to $4$ units, gear, and/or runes} to ready both empowered copies of \jcard{Platewrym Egg}, \jcard{Heimerdinger, Inventor}, and \jcard{Seal of Insight}. \label{AG: 2PWE, L, SoI}
        \item By having \jcard{Heimerdinger, Inventor} copy \jcard{Jayce, Defender of Tomorrow}'s exhaust ability, we may repeat \ref{JT4}.(\ref{EEPWE}) and \ref{JT4}.(\ref{EEJL+EEPWE}) using \jcard{Heimerdinger, Inventor} instead to add a total of $7$ energy. \label{HeimerFix}
        \item Due to the discount of \jcard{Applied Researchers} spend $4$ energy to play \card{Progress Day}{OGN 114: Draw 4}. ($3$ energy remains) \label{PlayPD}
        \item Exhaust \jcard{Seal of Insight} and spend $1$ energy to repeat \ref{JT5}.\ref{PlayPrem} ($2$ energy remains).
        \item Due to the discount of \jcard{Applied Researchers} spend $1$ energy to play \card{Dredge Up}{VEN 049: Draw 1. / Flow: 2 Energy}. ($1$ energy remains) \label{PlayD}
        \item Play the second \jcard{Acceleration Gate}; repeat \ref{JT5}.(\ref{AG: 2PWE, L, SoI}), \ref{JT5}.(\ref{HeimerFix}), and \ref{JT5}.(\ref{AddMindPower}) to add $7$ energy and $1$ mind power. 
        \item  Repeat \ref{JT5}.(\ref{PlayPD}) and using the discount of applied researchers, spend $2$ to play \card{Consult the Past}{OGN 083: Hidden/ Reaction/ Draw 2}. (1 energy remains)
        \item With the remaining energy play the last copy of \jcard{Acceleration Gate};  repeat \ref{JT5}.(\ref{AG: 2PWE, L, SoI}) to then repeat \ref{JT4}.(\ref{EEPWE}), \ref{JT5}.(\ref{HeimerFix}), \ref{JT5}.(\ref{AddMindPower}) to add a net of $7$ energy and $1$ mind power.
        \item  Spend $2$ energy to repeat \ref{JT5}.(\ref{PlayD}) twice ($5$ energy and $1$ mind power remain).
        \item Using the discount of applied researchers, spend $2$ to play \jcard{Consult the Past} ($3$ energy and $1$ mind power remain).
        \item Repeat \ref{JT4}.(\ref{QT}) to draw the second to last card in the deck ($2$ energy and $1$ mind power remain).
        \item Exhaust a copy of \jcard{Hextech Formula} to empower the exhausted \jcard{Questionable Tome}. \label{EmpowerQT} and pass the turn.
    \end{enumerate}
        \subsubsection{Bob Turn 5} Under our assumptions Bob scores two points.
        \subsubsection{Jayce Turn 6 (Bob Score 7)}\label{JT6} During the draw step of this turn, Jayce draws the last card of the deck.
    \begin{enumerate}
        \item Exhaust $2$ runes and spend the $2$ energy created to play \card{Garbage Grabber}{OGN 099: Recycle $3$ cards from your trash, $1$ Energy, Exhaust: Draw $1$}.
        \item Exhaust a rune to add an energy.
        \item Using the exhaust ability of \jcard{Garbage Grabber} recycle the $3$ copies of \jcard{Acceleration Gate} in the trash, and pay $1$ to draw a copy of \jcard{Acceleration Gate}. \label{GG-AG}
        \item Exhaust the rest of the available runes. ($3$ energy remains).
        \item Repeat \ref{JT4}.(\ref{EEPWE}), \ref{JT5}.(\ref{HeimerFix}), \ref{JT5}.(\ref{AddMindPower}) to add a net of $7$ energy and $1$ mind power. ($10$ energy and $1$ mind power remain).
        \item To finish the preliminary set-up for building Jayce's Computer, play the drawn copy of \jcard{Acceleration Gate} to ready both empowered copies of \jcard{Platewrym Egg}, \jcard{Heimerdinger, Inventor}, and \jcard{Seal of Insight}. ($9$ energy and $1$ mind power remain).
    \end{enumerate}

    Jayce's board state at this point in the algorithm is covered in Tables \ref{Table:JayceBase}-\ref{Table:Trash}.
    At this point Jayce can generate an unbounded amount of energy and mind power by repeating the following actions:
    \begin{enumerate}
        \setcounter{enumi}{6}
        \item Repeat \ref{JT4}.(\ref{QT}) to draw another acceleration gate. ($8$ energy and $1$ mind power remains) \label{DrawAG}
        \item Repeat \ref{JT5}.(\ref{AddMindPower})  to add a mind power. ($8$ energy and $2$ mind power remain)
        \item Repeat \ref{JT4}.(\ref{EEPWE}) to add $4$ energy. ($12$ energy and $2$ mind power remain)
        \item Spend an energy and exhaust \jcard{Heimerdinger, Inventor} to ready \jcard{Questionable Tome} and \jcard{Seal of Insight} ($11$ energy and $2$ mind power remain)
        \item Exhaust a copy of \jcard{Hextech Formula} to empower \jcard{Questionable Tome}. 
        \item For $1$ energy play \jcard{Acceleration Gate} 
        to ready a copy of \jcard{Hextech Formula}, both empowered copies of \jcard{Platewrym Egg} and \jcard{Heimerdinger, Inventor}. ($3$ energy and $1$ mind power remain)
        \item Repeat the first four steps of this loop to draw the final acceleration gate and end up with a net 12 energy and 3 mind power. There are now no cards in the draw pile.
        \item For $1$ energy play \jcard{Acceleration Gate} to ready \jcard{Garbage Grabber},  both empowered copies of \jcard{Platewrym Egg} and \jcard{Heimerdinger, Inventor}. ($11$ energy and $3$ mind power remain).
        \item Repeat the second and fifth steps of this loop to add another mind power and empower \jcard{Questionable Tome} ($11$ energy and $4$ mind power remain.)
        \item Spend 1 and Exhaust \jcard{Heimerdinger, Inventor} to ready \jcard{Hextech Formula} and \jcard{Seal of Insight}. ($10$ energy and $4$ mind power remain)
        \item Repeat the second and third steps of this loop to end up with $14$ energy and $5$ mind power.
        \item Repeat \ref{JT6}.(\ref{GG-AG}).
        \item To complete the loop repeat \ref{JT5}.(\ref{AG: 2PWE, L, SoI}) (12 energy and 5 mind power remain).
    \end{enumerate}
    The above loop adds $3$ energy and $4$ mind power each time it is completed. Now that we have demonstrated an unbounded energy and mind power loop we only need to do resource accounting for body power. The following construction routines will facilitate the build of an $M\in\text{TM}(7,4)$ in \textit{Riftbound}.
    
\subsection{Construction Routines and Technical Subroutines} \label{CR&TS}
For each of the following routines, we may assume that Jayce's board state is in the configuration described in Tables \ref{Jayce1}-\ref{Table:Trash} with the necessary additions from each additional subroutine as they are presented here.

\subsubsection{An Ordering on Might}\label{Orderings} We first introduce a might ordering to assist in the construction process. Let $m$ be the function that takes a unit as an input and outputs that units current might value.  We say 
\begin{equation}
\begin{cases}
    \text{ if } \left\lfloor\frac{m(\mathsf{M}_i)}{3} \right\rfloor \equiv_2 \left\lfloor\frac{m(\mathsf{M}_j)}{3} \right\rfloor \equiv_2 0 \text{ then } \mathsf{M}_i \leq_m \mathsf{M_j} \text{ if and only if } m(\mathsf{M}_i) \geq m(\mathsf{M}_j) \\
    \text{ if } \left\lfloor\frac{m(\mathsf{M}_i)}{3} \right\rfloor \equiv_2 \left\lfloor\frac{m(\mathsf{M}_j)}{3} \right\rfloor \equiv_2 1 \text{ then } \mathsf{M}_i \leq_m \mathsf{M_j} \text{ if and only if } m(\mathsf{M}_i) \leq m(\mathsf{M}_j)\\
    \text{ if }  0=\left\lfloor\frac{m(\mathsf{M}_i)}{3} \right\rfloor \not\equiv_2 \left\lfloor\frac{m(\mathsf{M}_j)}{3}\right\rfloor = 1 \text{ then } \mathsf{M}_i <_m \mathsf{M_j}
\end{cases}
\end{equation}

The strict ordering $<_m$ is defined in the natural way. In other words, suppose $m(\mathsf{M}_n) = 3n$. Then this ordering corresponds to a standard bijection from $\NN\setminus \{0\}$ to $\ZZ$ after integer division by $3$. So, 
\begin{equation}\label{mechorder}\cdots<_m\mathsf{M}_4<_m\mathsf{M}_2<_m\mathsf{M}_1<_m\mathsf{M}_3<_m\cdots.
\end{equation}
This choice of order is similar to the method of ordering tokens used in \cite{MagicCBH20}.

\subsubsection{Drawing Other Cards}\label{DOC} Some of the cards will be played multiple times in the computation and construction of Jayce's Computer. So, we need to describe a method for recycling and drawing cards other than \jcard{Acceleration Gate}. Given our assumption on the board state, this algorithm will draw three cards, $X$ and $Y$, and a copy of \jcard{Acceleration Gate} where at most one of $X$ or $Y$ may also be a copy of \jcard{Acceleration Gate}.

\begin{enumerate}
    \item Repeat \ref{JT6}.(\ref{DrawAG})
    \item Exhaust \jcard{Heimerdinger, Inventor} to ready \jcard{Questionable Tome} and \jcard{Garbage Grabber} if needed.
    \item Exhaust \jcard{Garbage Grabber} and recycle $X$, $Y$, and \jcard{Acceleration Gate} from the trash to draw \jcard{Acceleration Gate}. \label{DOCStart}
    \item Exhaust a copy of \jcard{Hextech Formula} to empower \jcard{Questionable Tome} then repeat \ref{JT4}.(\ref{QT}) to draw one of the three recycled cards. \label{Draw}
    \item Play \jcard{Acceleration Gate} to ready \jcard{Heimerdinger, Inventor}, both copies of \jcard{Hextech Formula} and \jcard{Questionable Tome}.
    \item Repeat \ref{DOC}.(\ref{Draw}) to draw the second card recycled.
    \item Exhaust \jcard{Heimerdinger, Inventor} to ready \jcard{Hextech Formula} and \jcard{Questionable Tome}.
    \item Repeat \ref{DOC}.(\ref{Draw}) to draw the third card recycled.
    \item Exhaust \jcard{Hextech Formula} to empower \jcard{Questionable Tome}, then play \jcard{Acceleration Gate} to ready both copies of \jcard{Hextech Formula}, \jcard{Heimerdinger}, and \jcard{Garbage Grabber} if needed.\label{DOCEnd}
\end{enumerate}

At the end of one iteration, one has a copy of \jcard{Acceleration Gate}, $X$, and $Y$ in hand as desired. From here it is easy to return to the original board state, or one with more options available by playing \jcard{Acceleration Gate}. Sometimes one may want to reserve this copy of \jcard{Acceleration Gate} when the draw other cards construction routine is used within another routine. 

\subsubsection{Building the Tape and Representing Symbols}\label{T&S} Jayce will use \jcard{Mech} tokens for the tape. The state of these \jcard{Mech} tokens will also represent a symbol. A \jcard{Mech}, $\mathsf{M}$ is $0$ (or the empty symbol), a buffed \jcard{Mech} represented by $\mathsf{M}^b$ is $1$, an empowered \jcard{Mech} represented by $\mathsf{M}^e$ is $2$, and an empowered and buffed \jcard{Mech}, $\mathsf{M}^{be}$, is a $3$. We now describe how to generate an unbounded number of \jcard{Mech} tokens.
\begin{enumerate}
    \item Since our decklist has $5$ body runes and $7$ mind runes, it is possible that, of the $6$ runes available, all of them are \jcard{Mind Rune} cards. If this is the case, play \card{Catalyst of Aeons}{OGN 138: Channel $2$ runes exhausted. ...} to channel $2$ runes. The pigeonhole principle guarantees that there is a \jcard{Body Rune} available to recycle for the next step.
    \item Recycle a \jcard{Body Rune} for $1$ body power and use it to play \card{Hextech Disc}{VEN 087: Empower: Exhaust / Disempower this, $1$ Energy, Exhaust: Play a $3$ might \jcard{Mech} unit token to your base}. \label{preloop}
    
\end{enumerate} 
    The first such \jcard{Mech} token created in the following loop will be $\mathsf{M}_1$.
\begin{enumerate}[start = 3]
    \item If there was already a \jcard{Mech} token $\mathsf{M_{i-1}}$ on the board when a \jcard{Mech} $\mathsf{M_{i}}$ is created, using \ref{DOC}, we repeatedly play \card{Grim Resolve}{UNL 095: Action/ Give a friendly unit +3 might this turn. ...} $(i-1)$ times targeting $\mathsf{M}_i$ \label{mechmight}
    \item Exhaust the second copy of \jcard{Hextech Formula} to empower \jcard{Hextech Disc}. \label{T&SStart}
    \item Disempower and Exhaust \jcard{Hextech Disc} to make a \jcard{Mech} token. Repeat \ref{T&S}.(\ref{mechmight})
    \item Repeat \ref{JT6}.(\ref{DrawAG}) then play the drawn \jcard{Acceleration Gate} to ready \jcard{Hextech Disc}, both copies of \jcard{Hextech Formula} and \jcard{Questionable Tome}.
    \item Exhaust one \jcard{Hextech Formula} to empower \jcard{Questionable Tome} then repeat \ref{JT6}.(\ref{DrawAG}) to draw the last \jcard{Acceleration Gate}.\label{T&SResetStart}
    \item Play the last \jcard{Acceleration Gate} to ready \jcard{Garbage Grabber} (if needed), \jcard{Questionable Tome} and a copy of \jcard{Hextech Formula}.
    \item Repeat \ref{JT6}.(\ref{GG-AG}). Then play \jcard{Acceleration Gate} to ready \jcard{Garbage Grabber} \label{T&SEnd}
    \item Repeat steps \ref{T&S}.(\ref{T&SStart})-\ref{T&S}.(\ref{T&SEnd}) as needed.
\end{enumerate}

Before constructing the tape as above, we will play \card{Arena Bar}{OGN 124: Exhaust: Buff an exhausted friendly unit.}. With this in play, we can write a $1$ or $2$ at position $i$ as we construct the tape. Note that we may write $3$ by writing both $1$ and $2$ at position $i$. We write $1$ and $2$ in the following way:
\begin{itemize}
    \item[] 1: \begin{enumerate}
        \item Exhaust Arena Bar to buff $\mathsf{M_i}$ (note that this does not change the order on the tokens).
        \item During \ref{T&S}.(\ref{T&SEnd}) ready \jcard{Arena Bar}
    \end{enumerate}
    \item[] 2: Perform \ref{TMStates}.(\ref{RHHTFHTF})-\ref{TMStates}.(\ref{RHTrueEnd}) but instead of empowering $O_i$ empower $\mathsf{M}_i$.
\end{itemize}

\subsubsection{The Read Head}\label{RH} Jayce will use the total number of \jcard{Sprite} unit tokens in his base for the read head. If there are $n$ \jcard{Sprite} tokens on the board, the read head is positioned at $\mathsf{M}_n$. We can achieve this by playing \card{Sprite Call}{OGN 094: Play a ready 3 might sprite unit token with temporary.} as many times as needed using \ref{DOC} to draw it over and over again. We will cover moving the read head in more detail in \ref{Transitions}.


\subsubsection{Representing a member of \normalfont{TM}$(7,4)$ \textit{States}}\label{TMStates} To embed the $7$ states of $M\in\text{TM}(7,4)$ into Jayce's Computer, Jayce will use $3$ gear $\{\mathsf{G}_1,\mathsf{G}_2,\mathsf{G}_3\}$ that are unique in base and that do not have an empower or disempower ability. Each $\mathsf{G_i}$ should have an exhaust ability that only interacts with the resources available and does not draw cards. \jcard{Ancient Henge}, \jcard{Hextech Anomaly} and \jcard{Seal of Insight} are designated for this purpose in our decklist. \jcard{Ancient Henge} and \jcard{Hextech Anomaly} both have an exhaust ability regarding resource conversion and may make other routines and subroutines more efficient depending on what resources are available in the rune pool, but we do not explicitly need their use. Also Jayce will use $3$ other game objects $\{O_1,O_2,O_3\}$ that are paired with the gears and do not interact with the empowered state. Both copies of \jcard{Jayce, Brilliant Inventor} and \jcard{Garbage Grabber} work for this purpose after using \jcard{Arena Bar} to buff one of the copies of \jcard{Jayce, Brilliant Inventor} to distinguish the two. 

To represent each of the $7$ states, Jayce represents the state's index in binary using the empowered ($1$) or disempowered ($0$) state of the  $\mathsf{G}_i$. To do so, Jayce first empowers the $O_i$ by repeating the following steps for the different values of $i$.

\begin{enumerate}
    \item Exhaust \jcard{Hextech Formula} to empower the other copy of \jcard{Hextech Formula}\label{RHHTFHTF}
    \item Play \card{Profiteer}{VEN 082: When you play me, you may disempower something you control to empower a legend, unit, or gear}. to disempower \jcard{Hextech Formula} to empower $O_i$.
    \item Play \jcard{Singularity} to kill \jcard{Profiteer}.\label{KillProf}
    \item Apply \ref{DOC} with $X = $\jcard{Profiteer} and $Y= $\jcard{Singularity}.\label{DrawProf}
    \item Play \jcard{Acceleration Gate}, the one drawn from applying \ref{DOC}, to ready both copies of \jcard{Hextech Formula} and \jcard{Garbage Grabber}.\label{RHTrueEnd}
\end{enumerate}


Now, if $\mathsf{G}_i$ should be empowered, Jayce plays \jcard{Profiteer} targeting $O_i$ to empower $\mathsf{G}_i$. If $\mathsf{G}_i$ should be disempowered they target $\mathsf{G}_i$ to empower $O_i$ instead.

We also want to have substates that determine if Jayce's Computer is transitioning or writing to the tape. We say that the embedded $q_i$ state from our $M\in\text{TM}(7,4)$ corresponds to the binary representation described above where all $\mathsf{G_i}$ are \textit{ready}. When a transition requires writing to the tape or changing state, we exhaust an appropriate $\mathsf{G_i}$. More details follow in \ref{Transitions}.

\subsection{The Transition Function}\label{Transitions} In the following, the \jcard{Mech} tokens making up the tape have been arranged according to the might ordering $<_m$ as shown in \eqref{mechorder}. Suppose that the transition function $\delta$ is given and consider $\delta(s,q_l) = (s',q_j,D)$ where $s,s'\in \{0,1,2,3\}$ are symbols represented by $\{\mathsf{M},\mathsf{M}^b,\mathsf{M}^{e},\mathsf{M}^{be}\}$, $q_l,q_j\in \{q_i\}_{i\in\{1,2,\dots, 7\}}$ are states, and $D\in \{L,R\}$ is the direction the read head will move. 
\begin{enumerate}
\item To signal that a transition has started, if needed perform \ref{TMStates}.(\ref{RHHTFHTF}) to empower a copy of \jcard{Hextech Formula}. Then play \jcard{Profiteer} to disempower \jcard{Hextech Formula} and empower a \jcard{Sprite} token. 
\item Perform \ref{TMStates}.(\ref{KillProf}) and \ref{TMStates}.(\ref{DrawProf}) to kill \jcard{Profiteer} with \jcard{Singularity} and then redraw \jcard{Singularity}, \jcard{Profiteer}, and a copy of \jcard{Acceleration Gate}. \label{Profloop}
\item Play \jcard{Acceleration Gate} to put it in the trash. \label{TrashAG}
\item If $q_l \neq q_j$, exhaust \jcard{Hextech Anomaly}. 
\item If $s \neq s'$, exhaust \jcard{Ancient Henge}.
\item If \jcard{Hextech Anomaly} is exhausted:
\begin{enumerate}
    \item Write $l$ and $j$ in binary as $l_1l_2l_3$ and $j_1j_2j_3$ respectively. Note that $l_1l_2l_3$ determine the current empower-disempower states of the pair $(O_n,\mathsf{G}_n)$ for $n\in \{1,2,3\}$.
    \item For each $n\in \{1,2,3\}$ check if $l_n = j_n$. 
    \begin{itemize} 
    \item If true, then do nothing. 
    \item If false, then play \jcard{Profiteer} to flip the empower-disempower states of $O_n$ and $\mathsf{G}_n$ then do \ref{Transitions}.(\ref{Profloop}).
    \end{itemize}
\end{enumerate}
\item If \jcard{Ancient Henge} is exhausted:
\begin{enumerate}
    \item Note that $s$ is represented by $\mathsf{M}^k$ and the new symbol only changes $k$. Write $k = k_1k_2$ and $k' = k'_1k'_2$ both as elements of $\{\text{\textvisiblespace}\text{\textvisiblespace}, \text{\textvisiblespace}e, b\text{\textvisiblespace}, be\}$.
    \item If $k_1\neq k'_1$:
    \begin{itemize}
        \item If $k'_1= \text{\textvisiblespace}$: 
        \begin{enumerate}
            \item Play \jcard{Profiteer}
            \item Using the buff on $\mathsf{M}_j$ play \card{Wallop}{OGN 146: Action/ As you play this, you may spend a buff as an additional cost. If you do, ignore this spell's cost. / Ready a unit.} to ready \jcard{Profiteer}.
            \item Do \ref{Transitions}.(\ref{Profloop}) and \ref{Transitions}.(\ref{TrashAG}).
            \item Perform \ref{DOC} with $X = $ \jcard{Wallop} and $Y = $ \jcard{Acceration Gate}
            \item Do \ref{Transitions}.(\ref{TrashAG}) twice.
        \end{enumerate}
        \item If $k'_1 = b$:
        \begin{enumerate}
            \item Exhaust \jcard{Arena Bar} to buff $\mathsf{M}_j$.
            \item Perform \ref{JT6}.(\ref{DrawAG}) then play the drawn \jcard{Acceleration Gate} to ready \jcard{Arena Bar}, both copies of \jcard{Hextech Formula} if needed and \jcard{Questionable Tome}.
            \item Perform \ref{T&S}.(\ref{T&SResetStart})-\ref{T&S}.(\ref{T&SEnd})
        \end{enumerate}
    \end{itemize}
    \item If $k_2\neq k_2'$
    \begin{itemize}
        \item If $k_2' = \text{\textvisiblespace}$ 
        \begin{enumerate}
            \item Play \jcard{Profiteer} to disempower $\mathsf{M}_j$ and empower \jcard{Profiteer}. 
            \item Do \ref{Transitions}.(\ref{Profloop}) and \ref{Transitions}.(\ref{TrashAG}).
        \end{enumerate}
        \item If $k_2' = e$
        \begin{enumerate}
            \item If needed do \ref{TMStates}.(\ref{RHHTFHTF}).
            \item Play \jcard{Profiteer} to disempower \jcard{Hextech Formula} and empower $\mathsf{M}_j$.
            \item Do \ref{Transitions}.(\ref{Profloop}) and \ref{Transitions}.(\ref{TrashAG}).
        \end{enumerate}
    \end{itemize}
\end{enumerate}
    \item Note the number $r$ of \jcard{Sprite} tokens in base.
    \begin{itemize}
            \item If $r\equiv_2 1$ and $D = R$, or , if $r\equiv_2 0$ and $D = L$: 
                \begin{enumerate}
                    \item Play \jcard{Sprite Call} 
                    \item Use \ref{DOC} with $X = $\jcard{Sprite Call} and $Y = $\jcard{Acceleration Gate}.
                    \item Play \jcard{Sprite Call} and do \ref{Transitions}.(\ref{TrashAG}). twice.
                \end{enumerate}
            \item If $r > 1$, $r\equiv_2 1$ and $D = L$, or if $r>2$, $r\equiv_2 0$ and $D = R$:
                \begin{enumerate}
                    \item Play \jcard{Singularity} to kill the empowered \jcard{Sprite} token and another \jcard{Sprite} token
                    \item Use \ref{DOC} with $X = $\jcard{Singularity} and $Y = $\jcard{Acceleration Gate}.
                    \item Do \ref{Transitions}.(\ref{TrashAG}). twice.
                \end{enumerate}
            \item If $r=1$ and $D = L$:
                \begin{enumerate}
                    \item Play \jcard{Sprite Call} 
                    \item Use \ref{DOC} with $X = $\jcard{Sprite Call} and $Y = $\jcard{Acceleration Gate}.
                    \item Do \ref{Transitions}.(\ref{TrashAG}). twice.
                \end{enumerate}
            \item If $r = 2$ and $D = R$:
                \begin{enumerate}
                    \item Play \jcard{Singularity} to kill the empowered \jcard{Sprite} token.
                    \item Use \ref{DOC} with $X = $\jcard{Singularity} and $Y = $\jcard{Acceleration Gate}.
                    \item Do \ref{Transitions}.(\ref{TrashAG}). twice.
                \end{enumerate}
    \end{itemize}
    \item If there is an empowered \jcard{Sprite} token. Play \jcard{Profiteer} to disempower the empowered \jcard{Sprite} token and empower \jcard{Profiteer}, then do \ref{Transitions}.(\ref{Profloop}).
    \item Play \jcard{Acceleration Gate} to ready \jcard{Hextech Anomaly}, \jcard{Ancient Henge}, and also \jcard{Seal of Insight} if needed.
\end{enumerate}
\vspace{-.175cm}At this point, the transition has finished ending in an embedded state of a Turing machine $M\in\text{TM}(7,4)$ since all $\mathsf{G}_i$ are ready. 

\subsubsection{Halting}\label{halting}
When implementing UTM$(7,4)$, we may encode halting by making Jayce win the game. When the machine reads the halt symbol, Jayce will play \card{Renata Glasc, Mastermind}{SFD 088: ... / $4$ energy $4$ mind power, Exhaust: Score $1$ point/ ... / Use my abilities only while I'm at a battlefield.}. Then Jayce will ready and exhaust \jcard{Heimerdinger, Inventor} repeatedly to use the scoring exhaust ability of \jcard{Renata Glasc, Mastermind} until Jayce wins. There are several methods available for repeatedly readying \jcard{Heimerdinger, Inventor} in the algorithms here! Alternatively if we are looking for an accept or reject state instead of the halt symbol of the UTM, we may force Jayce to win in the same way when accepting. When rejecting, we may force Jayce to lose by continuing to attempt to draw from an empty deck by a slight alteration to \ref{DOC}. 

\section{Jayce's Strategy and Conclusions}\label{strategy}
Given a TM input pair $(M,w)\in$TM$(7,4)\times\{0,1,2,3\}^*$ Jayce proceeds as follows:
\begin{enumerate}
    \item Construct UTM$(7,4)$ as detailed in the algorithms above. If at any point the construction is interrupted by an opponent's action that would change might of, kill, or remove cards in your base, concede.
    \item Encode $w$ onto the tape using Section \ref{T&S}.
    \item Simulate $M$ on $w$ using Section \ref{Transitions}.
    \item If $M$ halts on $w$, use Section \ref{halting} to win or lose accordingly.
\end{enumerate}

This strategy is well-defined as a UTM$(7,4)$ exists by \cite[Section 6.]{RS96}. This strategy shows that \textit{Riftbound} is Turing complete. Further, we have shown that it is undecidable that a strategy in \textit{Riftbound} has a winning turn, proving Theorem \ref{MainTheorem}. Setting up the machine in this strategy is quite involved from the point of view of the number of cards played. We can be more precise about this by stating the following proposition.

\begin{proposition}
    For any $M\in \text{TM}(7,4)$ Jayce's Computer has string-card complexity $O(k^2)$.
\end{proposition}

\begin{proof}
     Prior to \ref{T&S}.\eqref{preloop}, the cards played require some constant $C$ of resources. Suppose the process involved to set up $k$ slots on the tape require $f(k)$ cards. Then since the energy and mind power loop generates $3$ energy and $4$ mind power every $3$ cards played ($3$ copies of \jcard{Acceleration Gate} are played in each repetition of the energy and mind power loop), the total number of cards required accounting for energy costs is bounded above by $O(f(k))+\frac{C}{3}$. So without loss of generality we may restrict our attention to the total number of cards played in the process. The careful reader will note that \ref{T&S}.\eqref{mechmight} is the most card-intensive step in the set up of the tape, the repeated play of \jcard{Grim Resolve} requires \ref{DOC}. So for the $i^{th}$ \jcard{Mech}, $O(i)$ cards are played. Taking into account the other steps in \ref{T&S} are $O(1)$ for the $i^{th}$ \jcard{Mech}, a quick calculation shows that the string-card complexity is 
    $$
        O\left( \sum_{i=1}^{k} O(i)\right) = O\left(O\left(\frac{k(k+1)}{2}\right)\right) = O(k^2)
    $$
\end{proof}

\section{Problems}\label{C&P}

Theorem \ref{MainTheorem} places \textit{Riftbound} in a similar category as \textit{Magic} in terms of its complexity, however it is important to note that our result is weaker as the strategy for constructing Jayce's Computer does not force the moves of the opponent. As mentioned earlier, there are cards that do force player action and so we arrive at a natural question:
\begin{question}
    Does there exist a construction of a UTM in \textit{Riftbound} in which the actions of all players are forced?
\end{question}
With the current card pool we suspect the answer to this question is no, but we may be wrong and there may be a clever solution that we are unaware of. Alternatively, as Jayce's Computer cannot handle interaction from the opponent, a weaker but still natural question in this same vein would be:
\begin{question}
    Does there exist a construction of a UTM in \textit{Riftbound} that prevents the opponent from interrupting its construction and computation?
\end{question}
There are legal cards in the card pool (at the time of writing) that restrict or stop interaction like \card{Lilting Lullaby}{UNL 190} and \card{Brynhir Thundersong}{OGN 026} respectively. As \jcard{Brynhir Thundersong} is in the fury domain, there are known infinite move combos here that use little resources that may be modified to also allow for the generation of infinite resources perhaps with the inclusion of \card{Jhin, Murderous Artist}{UNL 022}, but we have not thought about this problem in further detail.

Here we would also like to note that Jayce's Computer is not the only machine that we constructed in the development of this project, it was however, the easiest to keep track of as it only requires interaction with one zone of the game --- Jayce's base. There is actually a valid sequence of game states in \textit{Riftbound} which simulate any given Turing machine. For this more general case, one player plays the \card{Ivern, Green Father}{UNL 195} legend, and the other plays \card{Miss Fortune, Bounty Hunter}{OGN 267} legend with the following cards included in their decks.

 \subsubsection{Ivern's Decklist Inclusions}
    \begin{itemize}
        \item \jcard{Vanguard Armory}
        \item \jcard{Tianna Crownguard}
        \item \jcard{Lee Sin Astetic}
        \item \jcard{Rally the Troops}
        \item \jcard{Needlessly Large Yordle}
        \item \jcard{Imperial Decree}
        \item \jcard{Find Your Center}
        \item\jcard{Forge the Future}
        \item \jcard{Zaun Punk}
        \item \jcard{Charm}
    \end{itemize}
    \subsubsection{Miss Fortune's Decklist Inclusions}
    \begin{itemize}
        \item \jcard{Akshan}
        \item \jcard{Possesion}
        \item \jcard{Ride the Wind}
        \item \jcard{Flash}
        \item \jcard{Star-Crossed}
        \item \jcard{Baron Nashor}
        \item \jcard{Disposal Order}
        \item \jcard{Pack of Wonders}
        \item \jcard{Bewtiching Spirit}
    \end{itemize}

However, this more general construction is weaker than Jayce's Computer because it requires the game-state-directed strategies of the players to conspire to build and run the given machine and uses all zones of the game. This construction also has a larger turn complexity than Jayce's Computer for deciders. We leave this cooperative construction of an arbitrary Turing machine in \textit{Riftbound} as an exercise for the reader, and perhaps, their friend. 

\begin{exercise}
Define two game-state-directed strategies for players playing \jcard{Ivern, Green Father} and \jcard{Miss Fortune, Bounty Hunter} whose interaction constructs a given Turing machine $M$ and simulates a given input $w$ on that machine such that the Miss Fortune wins when the machine accepts and the Ivern wins when the machine rejects.
\end{exercise}

It would also be interesting to improve the string-card complexity of Jayce's Computer while retaining its turn complexity.

\begin{question}
    Does there exist a game-directed-strategy in \textit{Riftbound} that constructs a given Turing machine $M$ where strings of length $k$ have string-card complexity $O(k^n)$ for $n<2$ and where computation has turn complexity $O(1)$ if $M$ is a decider?
\end{question}

Further it we suspect that Jayce's Computer may be implemented faster than as described above:

\begin{question}
    Does there exist a game-direceted-strategy in \textit{Riftbound} for Jayce's Computer where the construction of the Turing machine occurs on Jayce's fifth turn or sooner?
\end{question}

\section{Appendix: How to play \textit{Riftbound: League of Legends Trading Card Game}}\label{Append}
 We will only cover rules and ideas relevant to constructing Jayce's Computer in the game. For more in depth information on rules of \textit{Riftbound} see the core rules \cite{CR} and the tournament rules \cite{TR}.\footnote{This project was completed just after the release of the \textit{Vendetta} set of \textit{Riftbound}, for future readers, more up-to-date rules can be found at \url{https://playriftbound.com/en-us/rules-hub/}.}  We summarize the card types in \textit{Riftbound} here.
    \subsection{\textit{Card Types}}

        \subsubsection{Legend} Cards that define the overall deck. \jcard{Jayce, Defender of Tomorrow} is one such card. These cards usually have additional effects or activated abilities. A deck's chosen champion unit must match the identity of the deck's legend.
        \subsubsection{Unit} Essentially, units are cards with an additional might indicator. Units may be played to any location under that player's control. Unless otherwise stated or restricted, units enter exhausted and all have intrinsic abilities ``Exhaust me: move me from base to a battlefield" and ``Exhaust me: move me from my battlefield to base" unless otherwise stated. A units' might functions both as its hit point value and its attack power. When a unit takes damage that meets or exceeds its might value, that unit is killed and goes to the trash. 
        \subsubsection{Spell} Cards that are played to add to or start a chain. Spell cards can only be played on a players turn unless they have the action or reaction keywords or some other alternative play window indicator where certain other rules apply. Actions can start chains, reactions can start or add to chains. 
        \subsubsection{Gear} Cards that are played to base that either have persistent effects or have abilities that can be activated. Gear enter ready unless otherwise stated.  
        \subsubsection{Battlefield} Cards that designate locations to conquer and hold throughout the game. Some battlefields may change or restrict aspects of the game.
        \subsubsection{Rune} Cards in the rune deck. These can be exhausted to add energy. They may also be recycled back into the rune deck to generate power of the domain stated on the rune card. 

    \subsection{Deck Construction and Set Up}
    \textit{Riftbound} is played with decks built around a legend card. The main deck consists of $40$ gear, spell, and unit cards, one of which is a chosen champion matching the legend card for the overall deck. The rune deck consists of $12$ runes matching the domains of the legend card. A player also includes a number of unique battlefields  depending on the game mode. These battlefields are kept separately from the main and rune decks.

    Each player sets their legend and chosen champion aside, chooses a battlefield for the game, and shuffles both the other $39$ cards of their main deck and their rune deck separately. Each player starts by drawing a hand of $4$ cards from their main deck and may mulligan up to $2$ of the $4$ cards by recycling them to the bottom of the main deck before drawing to replace them.

    \subsection{Tokens} Tokens are game objects generated by effects and abilities of other cards. Tokens are not played, but instead created and destroyed. They can only exist at base or at battlefields (if unit tokens), or be a location themselves (if battlefield tokens). If they are banished, killed, or would be put in trash they are instead destroyed. At the time of writing, there are only status indicator, unit, gear, and battlefield tokens in the game. Each token is considered a game object.

    \subsection{Exhaust and Ready} The exhausted state of a card is denoted by a card on the board horizontally, the ready state of a card is denoted by a card on the board upright. Various abilities and effects may ready and exhaust units, some abilities require that a card is exhausted and some may require that a card is ready. Units, Gear, and Legends can all be exhausted or readied. 

    \subsection{Energy and Power} The main resource pool in \textit{Riftbound}. Energy is domain-less and is generated by exhausting runes of any domain. Power on the other hand has an associated domain and is generated by recycling runes of that domain. There are also cards that have abilities that generate resources. Some power generation abilities generate a domain-less power that can be used for any domain power cost. 

    \subsection{Abilities and Chains} Abilities generally have a cost or trigger associated to them. When the cost is paid or the trigger condition is met, abilities get added to a chain. When a chain has been added to, their is a window for players to play cards with the reaction keyword to add to the chain. When a chain resolves, items on the chain are resolved in order starting with the most recently added chain item.

    \subsection{Showdowns} Essentially showdowns start whenever a player moves to an open battlefield or a battlefield that they do not control. We only assume that the opponent starts open showdowns in the first few turns of the Jayce's Computer construction, and Jayce does not interact with battlefields for the entirety of the game so there are many nuances here that the interested reader should investigate should they be interested in learning more about \textit{Riftbound} in general.

    \subsection{Buff and Empower States} The buff and empowered states are binary states of cards in Riftbound with few exceptions.\footnote{\card{Lee Sin, Ascetic}{OGN 078} does not have a buff limit and  \card{Kayle, Justified}{VEN 134} can be empowered $3$ times.} It is important to reiterate here that tokens are game objects. Therefore, they may be empowered by \cite[Rule 123 and Rule 441]{CR}

    \subsection{Turn Structure} Turns in \textit{Riftbound} start with the following four initial steps.
    \begin{enumerate}
        \item Awaken: All cards under the player's control (unless otherwise stated) will ready if they are exhausted during this phase.
        \item Beginning: At the start of this phase, temporary effects trigger. This is also the phase where holding and any holding effects happen.
        \item Channel: During this phase $2$ (unless otherwise restricted or permitted by persistent effects that may be present) runes are ``channeled" ready from the rune deck. In simpler terms, the player flips the first two cards of the rune deck face up onto the board in the ready position. 
        \item Draw: The player draws a card for turn. 
    \end{enumerate}
    After the initial ABCD, \textit{Riftbound} has a considerably open ended turn structure with an action phase where the player may play as much as they are able to given they have enough resources available. A turn ends when a player passes the turn to their opponent.

    \subsection{Points, Point Tempo, and Winning the Game}\label{PointTempo} In \textit{Riftbound} the goal is to score $8$ points. Although there are battlefields that change this win condition, we do not consider them in this short paper and many of them have been banned in standard play. 
    
    Points are usually scored by conquering and/or holding battlefields. One conquers by winning a showdown at a battlefield, and holds by having control of a battlefield at the end of their beginning phase. One can score at a battlefield at most once per turn. Unless otherwise stipulated by card effects, scoring grants a point to the player. A player may also score if their opponent attempts to draw when their deck is empty. Some units have abilities that allow their controller to score points.

    Point tempo describes the progression of points for a player over turns in the game. The most common point tempos are $(0,1,3,5,7,8)$ and $(0,2,4,6,8)$ where the $0$ is potentially a place holder for several $0$'s. There are several other point tempos in the game, but for the purposes of this short paper we only consider a point tempo of (0,1,3,5,7,8). When Jayce goes first, Jayce's Computer would finish construction when the opponent has $7$ points in the above implementation.

\section{Appendix: Decklist and Tables}\label{D&T}
In the tables that follow, the main decklist for Jayce's Computer is separated into each location cards find themselves when Jayce is ready to perform the construction. The utility of each card after set-up is also detailed if it is possible to draw and play that card again. For those cards in Jayce's base, the exhaust/empowered states of the cards are also specified. We used $7$ mind runes and $5$ body runes and we do not specify battlefields since the deck only interacts with Jayce's base.
\clearpage
    \begin{table}[hbt!]
    \caption{Jayce's base when machine construction can\label{Jayce1} start.}\label{Table:JayceBase}
    \begin{center}
    \begin{tabular}{|c|c|c|c|c|}
    \hline
        Card & Type & Exhaust & Empowered &Utility\\
    \hline
        Applied Researchers & U & No & Yes & Discounts spells to a minimum of $1$.\\
    \hline
        Applied Researchers & U & No & Yes & Discounts spells to a minimum of $1$.\\
    \hline
        Garbage Grabber & G & Yes & No & Recycles trash. Card draw.\\
    \hline
        Heimerdinger, Inventor & U & No & Yes & Readys gear, $q_{\text{accept}}$ resolution.\\
    \hline
        Hextech Formula & G & No & No & Empowers gear.\\
    \hline
        Hextech Formula & G & No & No & Empowers gear.\\
    \hline
        Jayce, Brilliant Inventor & U & No & No & Machine State\\
    \hline
        Jayce, Brilliant Inventor & U &  No & No &Machine State\\
    \hline
        Jayce, Defender of Tomorrow & L & No & Yes & Readys Gear.\\
    \hline
        Platewrym Egg & G & No & Yes & Generates Energy.\\
    \hline
        Platewrym Egg & G & No & Yes & Generates Energy.\\
    \hline
        Questionable Tome & G & No & Yes & Card Draw. $q_{\text{reject}}$ resolution.\\
    \hline
        Seal of Insight& G & No & No & Generates Mind Power. Machine State.\\
    \hline
    \end{tabular}
    \end{center}
    \end{table}
    \begin{table}[h!]
    \caption{Jayce's deck when machine construction can start}\label{Table:Deck}
    \begin{center}
    \begin{tabular}{|c|c|c|}
    \hline
        Card & Type& Utility\\
    \hline
        Acceleration Gate& S& Perpetuates Turn  \\
    \hline
        Acceleration Gate& S& Perpetuates Turn \\
    \hline
    \end{tabular}
    \end{center}
    \end{table}
    \begin{table}[h!]
    \caption{Jayce's cards in hand when machine construction can start}\label{Table:Hand}
    \begin{center}
    \begin{tabular}{|c|c|c|}
    \hline
        Card & Type& Utility\\
    \hline
        Ancient Henge& G& Machine state.\\
    \hline
        Arena Bar& G& Writing to tape. \\
    \hline
        Catalyst of Aeons& S& Body rune search. \\
    \hline
        Garbage Grabber & G & Recycles trash. Card draw.  \\
    \hline
        Grim Resolve& S& Tape Generation.\\
    \hline
        Hextech Anomaly& G& Machine state.\\
    \hline
        Hextech Disc& G& Tape generation.\\
    \hline
        Sprite Call& S& Tape generation.\\
    \hline
        Mobilize& S& Body rune search.\\
    \hline
        Profiteer& U& State changes. Writing to tape. Moving read head.\\
    \hline
        Renata Glasc, Masterminding& U& $q_{\text{accept}}$ resolution.\\
    \hline
        Singularity& S& Unit replay.\\
    \hline
        Wallop & S& Writing to tape.\\
    \hline
    \end{tabular}
    \end{center}
    \end{table}
    \begin{table}[h!]
    \caption{Jayce's trash when machine construction can start}\label{Table:Trash}
    \begin{center}
    \begin{tabular}{|c|c|c|}
    \hline
        Card & Type & Utility\\
    \hline
        Acceleration Gate& S & Perpetuates Turn \\
    \hline
        Consult the Past& S & Card draw, $q_{\text{reject}}$ resolution. \\
    \hline
        Consult the Past& S & Card draw, $q_{\text{reject}}$ resolution.\\
    \hline
        Dredge Up& S & $q_{\text{reject}}$ resolution. \\
    \hline
        Dredge Up& S & $q_{\text{reject}}$ resolution. \\
    \hline
        Dredge Up& S & $q_{\text{reject}}$ resolution. \\
    \hline
        Premonition& S & $q_{\text{reject}}$ resolution. \\
    \hline
        Premonition& S & $q_{\text{reject}}$ resolution. \\
    \hline
        Progress Day& S & $q_{\text{reject}}$ resolution. \\
    \hline
        Progress Day& S & $q_{\text{reject}}$ resolution. \\
    \hline
        Sumpworks Map& G & Not used after setup \\
    \hline
        Sumpworks Map& G & Not used after setup \\
    \hline
    Sumpworks Map& G & Not used after setup \\
    \hline
    \end{tabular}
    \end{center}
    \end{table}
    
\newpage{}
\bibliographystyle{amsplain}
\bibliography{TCG.bib}

\end{document}